\documentclass{article}
\usepackage{graphicx}
\usepackage{sidecap}
\usepackage[utf8]{inputenc}
\usepackage[english]{babel}
\usepackage{amsthm}
\usepackage[toc,page]{appendix}
\usepackage{xcolor}
  \usepackage{graphicx}
  \usepackage{csquotes}
  \usepackage[colorlinks=true,urlcolor=blue]{hyperref}
  \usepackage[
backend=biber,
style=numeric,
sorting=none
]{biblatex}
\title{Bibliography management: \texttt{biblatex} package}
    \usepackage[T1]{fontenc}
    \usepackage[textwidth = 15cm]{geometry} 
    \usepackage{lmodern, amssymb}
    \usepackage{mathtools}
    \usepackage{braket, relsize, nccmath} 
    \usepackage{bigints}
    \usepackage[utf8]{inputenc}
\usepackage[english]{babel}
 \usepackage[utf8]{inputenc}
\usepackage[english]{babel}
\usepackage{amsthm}
\usepackage{amsmath}               
  {
      \theoremstyle{plain}
      
  }
\usepackage{amsmath}
   \usepackage{multirow}
   
   \usepackage[utf8]{inputenc}
\usepackage[english]{babel}
\providecommand{\keywords}[1]
{
  \small	
  \textbf{\textit{Keywords---}} #1
}

\title{Bibliography management: \texttt{biblatex} package}
\author{}
\date{ }

\numberwithin{equation}{section}
 
\theoremstyle{definition}
\newtheorem{definition}{Definition}[section]
    
\title{Tilted Nonparametric Regression Function Estimation}
\author{Farzaneh Boroumand$^{1,2}$, Mohammad T. Shakeri$^{2}$, Nino Kordzakhia$^{1,3}$,\\
Mahdi Salehi$^{4}$, Hassan Doosti$^{1*}$ \\
\small $^{1}$Department of Mathematics and Statistics, Faculty of Science and Engineering, Macquarie University, NSW, Australia.\\
\small(Farzaneh.Boroumand@mq.edu.au)\\
\small $^{2}$Department of Biostatistics, Health School,
Mashhad University of Medical Sciences, Mashhad, Iran.\\
\small(shakerimt@mums.ac.ir) \\
\small $^{3}$ (Nino.Kordzakhia@mq.edu.au)\\
\small $^{4}$Department of Mathematics and Statistics, University of Neyshabur, Neyshabur, Iran.\\
\small(salehi.sms@neyshabur.ac.ir)\\
\small $^{*}$ Correspond author: Hassan.Doosti@mq.edu.au  \\
}

\begin{document}
\maketitle
\begin{abstract}

This paper provides the theory about the convergence rate of the tilted version of linear smoother. We study tilted linear smoother, a nonparametric regression function estimator, which is obtained by minimizing the distance to an infinite order flat-top trapezoidal kernel estimator. We prove that the proposed estimator achieves a high level of accuracy. Moreover, it preserves the attractive properties of the infinite order flat-top kernel estimator. We also present an extensive numerical study for analysing the performance of two members of the tilted linear smoother class named tilted Nadaraya-Watson and tilted local linear in the finite sample. The simulation study shows that tilted Nadaraya-Watson and tilted local linear perform better than their classical analogs in some conditions in terms of Mean Integrated Squared Error (MISE). Finally, the performance of these estimators as well as the conventional estimators were illustrated by curve fitting to  COVID-19 data for 12 countries and a dose-response data set.

\end{abstract}

\section{Introduction}

Let the regression model be
\begin{equation}
    Y_i = r(X_i) + {\epsilon}_i, \; \; 1 \leq i \leq n,
\end{equation}
where  $(Y_1, X_1), (Y_2, X_2 ), … ,(Y_n, X_n)$, are the data pairs, the design variable $X\sim{f_X}$, $X$ and  $\epsilon$  are independent,  ${\epsilon}_i$'s are independent and identically distributed (iid) errors with zero mean   $E(\epsilon )=0$ and variance $E(\epsilon^2)=\sigma^2$.
The  regression function $r$ and $f_X$ are unknown.
In this paper, we will focus on a nonparametric approach to estimate $r$. The main subject of this study is a class of nonparametric estimators called linear smoother. Nadaraya-Watson estimator and local linear estimator are two prevailing members of this class of estimators.
An estimator $\breve{r}$  of $r$, is said to be a linear smoother if it can be written in a form of linear function of weighted $Y$ sample.
Let the weight-vector be
$$l(x)=(l_1(x),...,l_n(x))^T.$$
Then the linear smoother $\breve{r}$ can be written as  
\begin{equation}
   \breve{r}_n(x) = l(x)^TY= {\sum}_{i=1}^{n}{l_i(x)Y_i},
\end{equation}
where ${\sum}_{i=1}^{n}{l_i(x)}=1$, see Buja et al. \cite{buja1989linear}.
Nadaraya-Watson estimator and local linear estimator can be written as a form of linear smoother with the following weight functions.
The weight functions for Nadaraya-Watson smoother, see Nadaraya \cite{nadaraya1964estimating}, Watson  \cite{watson1964smooth}, are 
\begin{equation}
    l_{i,NW}(x)=\frac{K(\frac{X_i-x}{h})}{{\sum}_{j=1}^{n}{K(\frac{X_j-x}{h})}},\;\; i=1,\dots,n.
\end{equation}
For the standard local linear smoother the weight functions are defined as follows, 
\begin{equation}
    l_{i,ll}(x)=\frac{b_i(x)}{\sum_{j=1}^nb_j(x)},\; i=1,\dots,n, 
\end{equation}
\begin{equation}
  b_i(x)=K(\frac{X_i-x}{h})(S_n,_2(x)-(X_i-x)S_n,_1(x)), \; i=1,\dots,n,\nonumber  
\end{equation}
\begin{equation}
    S_n,_j(x)={\sum}_{i=1}^{n}{K(\frac{X_i-x}{h})(X_i-x)^j},\;  j=1, 2.\nonumber
\end{equation}
where $K$ is a kernel function. The kernel function depends on the bandwidth, or smoothing, parameter  $h$ and assigns weights to the observations according to the distance to the target point $x$, see McMurry and  Politis,  \cite{mcmurry2004nonparametric}.  The small values of $h$ cause the neighboring points of $x$ to have the larger influence on the estimate leading to curvature changes in the estimated curve. The larger values of $h$ imply that the distanced  data points will have the same effect as the neighboring points on the local fit, resulting in a smoother estimate. Thus  finding an optimal $h$ is the essential task in the estimation procedure, see Wasserman \cite{wasserman2006all}.
One of the ways finding the optimal $h$ is by minimising the leave-one-out cross validation score function,  \cite{wasserman2006all}.
The leave-one-out cross validation score is defined by   
\begin{equation}
    CV=\hat{R}(h)=\frac{1}{n} \sum_{i=1}^{n} (Y_i-\breve{r}_{(-i)}(X_i))^2,
\end{equation}
where $\breve{r}_{(-i)}(X_i)$ is obtained from (1.2) by omitting the $i^{th}$ pair $(X_i, Y_i)$.
In this work, we will present the tilted versions of linear smoother.   
A tilting technique applied to an empirical distribution, leads to replacing $1/n$ data weights from uniform distribution by $p_i, \; 1 \leq i \leq n,$ from general multinomial distribution over data. Hall and  Yao \cite{hall2003data} studied asymptotic properties of the tilted regression estimator with autoregressive errors using generalized empirical likelihood method, which typically involves solving a non-linear and high dimensional optimization problem. Grenander \cite{grenander1956theory} introduced a tilted method to impose restrictions on the density estimates. There are two approaches to estimating of the tilting parameters: Empirical likelihood  and Distance Measure based approaches. The empirical likelihood-based method is a semi-parametric method which provides a convenience of adding a parametric model through estimating equations. Owen \cite{owen1988empirical} proposed an empirical likelihood to be used as an alternative to the likelihood ratio tests, and derived its asymptotic distribution. Chen \cite{chen1997empirical}, Zhang \cite{biao1998note}, Schick et al. \cite{schick2009improved}, M{\"u}ller et al. \cite{muller2005weighted} further developed the empirical likelihood-based method for estimating the tilting parameters. Chen \cite{chen1997empirical} applied the empirical likelihood method to estimate the tilting parameters $p_i, \; 1 \leq i \leq n,$ under the constraints on the shape of distribution. In his kernel-based estimator, $n^{-1}$ was replaced by the weights obtained from the empirical likelihood method. In \cite{chen1997empirical} it was  proved that the proposed estimator has a smaller variance than the conventional kernel estimators. Schick et al.\cite{schick2009improved}, also used the similar  approach obtaining the consistent tilted estimator with  higher efficiency  than that of conventional estimators in the autoregression framework.
In contrast in the  Distance Measure approach, the tilted estimators are defined by minimizing distances, conditional to various types of constraints. Hall and Presnell \cite{hall1999intentionally}, Hall and Huang \cite{hall2001nonparametric}, Carroll et al. \cite{carroll2011testing}, Doosti and Hall \cite{doosti2016making}, Doosti et al. \cite{doosti2018nonparametric} used the setup-specific Distance Measure approaches for estimating the tilting parameters.
Carroll \cite{carroll2011testing},  proposed a new approach for density function estimation, and regression function estimation as well as  hypothesis testing under shape constraints in the model with measurement errors. A tilting method used in \cite{carroll2011testing}   led to curve estimators under some constraints. Doosti and Hall \cite{doosti2016making} introduced a new higher order nonparametric density estimator, using tilting method, where  they used $L_2$-metric between the proposed estimator and a consistent 'Sinc' kernel based estimator.  
Doosti et al. \cite{doosti2018nonparametric}, have introduced a  new way of choosing the bandwidth and estimating the tilted parameters based on the cross-validation function. In \cite{doosti2018nonparametric}, it was shown that the proposed density function estimator had improved  efficiency and was more cost-effective than the conventional kernel-based estimators studied in this paper. 

In this work, we propose a new tilted version of a linear smoother which is obtained by minimising the distance to a comparator estimator. The comparator estimator is selected to be an infinite order flat-top kernel  estimator. This class of estimators is characterized by a Fourier transform, which is flat near the origin and infinitely differentiable elsewhere, see \cite{mcmurry2008minimally}.  We prove that the tilted estimators achieve a high level of accuracy, yet preserving the attractive properties of an infinite-order flat-top kernel estimator.

The rest of this paper contains the additional four sections and Appendix. In the  Section 2, we provide the notation, definitions and preliminary results. The Section 2 also includes the definition of an infinite-order estimator, as a comparator estimator.  Section 3 contains the main results formulated in Theorems 1-3. We present a simulation study in the Section 4.  The real data applications are provided in the Section 5.  The proof of the main theorem is accommodated in the Appendix.

\section{Notation and preliminary results}
\theoremstyle{definition}
\begin{definition}
A general infinite order flat-top kernel $K$ is defined
 \begin{equation}
 K(x)=\frac{1}{2\pi}\bigints_{-\infty}^{\infty}\lambda(s)e^{-isx} ds,
 \end{equation}
where $\lambda(s)$ is the Fourier transform of kernel $K$, and $c>0$ is a fixed constant.
\begin{equation}
\lambda(s)=\begin{cases}
1, & \mid s\mid\leq c\\
g(\mid s\mid), & \mid s \mid >c 
\end{cases},\nonumber
\end{equation}
  and $g$ is not unique and it should be chosen to make $\lambda(s)$, $\lambda^2(s)$, and $s\lambda(s)$ integrable \cite{mcmurry2004nonparametric}. 
\end{definition}

\subsection{Infinite order flat-top kernel regression estimator }
Let $\check{r}$ be a linear smoother  
\begin{equation}
    \check{r} = {\sum}_{i=1}^{n}{\check{l}_i(x_i)Y_i},
\end{equation}
where
\begin{equation}
\check{l}_i(x)=\frac{K(\frac{X_i-x}{h})}{{\sum}_{j=1}^{n}{K(\frac{X_j-x}{h})}}\nonumber
\end{equation}
and $K$ is an infinite order flat top kernel from (2.1), also see McMurry and Politis \cite{mcmurry2004nonparametric}. The trapezoidal kernel 
\begin{equation}
    K(x)=\frac{2(cos(x/2)-cos(x))}{\pi x^2},\nonumber
\end{equation}
is an infinite order flat top kernel  satisfying Definition 2.1 since
the Fourier transform of $K(x)$ is 
\begin{equation}
\lambda(s)=\begin{cases}
1 &   \mid s\mid\leq 1/2,\\
2(1-\mid s\mid) &  1/2 < \mid s \mid\leq 1, \\ 0 & \mid s\mid>1.\nonumber
\end{cases}
\end{equation}
\subsection{Tilted linear smoother}
We define tilted linear smoother as follows
\begin{equation}
    \hat{r}_n(x|h,p) = {\sum}_{i=1}^{n}{p_i l_i(x_i)Y_i},
\end{equation}
where ${p_i}$'s are tilting parameters, $p_i \geq 0$ and $\sum_{i}^{n}p_i=1$.
The bandwidth parameter $h$ and the vector of tilting parameters $p=(p_1,\cdots,p_n)$, are to be estimated. 
In Section 4, we evaluate the performance of tilted versions of Nadaraya-Watson (1.2) and standard local linear estimators (1.3) in finite samples.
\section{Main results}
Let $\hat{r}_n(.|\theta)$ be the tilted linear smoother from (2.3) for the regression function ${r}$, where $\theta=(h,p)$ is a vector of unknown parameters. Further $\check{r}$ from (2.2) will be used as a comparator estimator of $r$, $\check{r}$ can be any estimator with an optimal convergence rate \cite{mcmurry2008minimally}. We will estimate $\theta$ by minimising the $L_2-$ distance  between $\hat{r}_n(.|\theta)$ and $\check{r}$ preserving the convergence rate of $\check{r}$, provided the following assumptions hold

\begin{enumerate}
\item [(a)] $\Vert \check{r} - r \Vert = O_p(\delta_n)$
\item [(b)] 
There exists  $\tilde{\theta}$  such that $\hat{r}_n(.|\tilde{\theta})$ and $\check{r}$ possess the same convergence rates, i.e. $\Vert\hat{r}_n(.|\tilde{\theta}) - r \Vert = O_p(\delta_n)$, 
\end{enumerate}
where  $\delta_n \geq 0$ converges to 0 as n tends to $\infty$, e.g. $\delta_n=n^{-c}$ for some $c\in (0,1/2)$.
A further discussion on the assumptions (a)-(b) can be found in Doosti and Hall, \cite{doosti2016making}.

We define $\hat{\theta}$ as the solution to the optimisation problem as 
\begin{equation}
    \hat{\theta}= \arg \min_{\theta} {\Vert\hat{r}_n(.|\theta) - \check{r} \Vert }
\end{equation}
subject to the constraints for the bandwidth parameter  $h>0$ and vector $p$ introduced in Section 2.2.
 
 In Theorem 1, we show that the convergence rate of $\hat{r}_n(.|\hat{\theta})$ and $\check{r}$ is $O_p(\delta_n)$.
\newtheorem{theorem}{Theorem}
\begin{theorem}
If the  assumptions (a)-(b) hold then for  any $\hat{\theta}$ which fulfills (3.1) we have 
$$
\Vert\hat{r}_n(.|\hat{\theta}) - r \Vert = O_p(\delta_n).
$$
\end{theorem}
\begin{proof}
Due to Assumption (a), there 
exists $\tilde{\theta}$ such that 
\begin{equation}
   \Vert\hat{r}_n(.|\tilde{\theta}) - \check{r} \Vert \leq \Vert\hat{r}_n(.|\tilde{\theta}) - r \Vert + \Vert r - \check{r} \Vert = O_p(\delta_n), \nonumber
\end{equation}
in which the first equation is a result of the triangle inequality, and specifically from the fact that
\begin{equation}
    \Vert r - \check{r} \Vert = O_p(\delta_n);
\end{equation}
see assumption, part(a). If $\Tilde{\theta}$ is as in assumption 1, part (2) then 
\begin{equation}
     \Vert\hat{r}_n(.|\hat{\theta}) - \check{r}\Vert \leq \Vert\hat{r}_n(.|\tilde{\theta}) - \check{r} \Vert = O_p(\delta_n).
\end{equation}
Together, results (3.2) and (3.3) imply Theorem 1.
\end{proof}

Theorem 1 implies that the convergence rate of $\hat{r}_n(.|\hat{\theta})$ estimator coincides with that of $\check{r}$ with the bandwidth parameter $h$ replaced by its ‘plug-in’ type estimate similar to that from \cite{mcmurry2004nonparametric} and \cite{mcmurry2008minimally}.

The regression function  $r\in \mathcal{C}$, where $\mathcal{C}$ is a class of regression functions, if 
 \begin{equation}
     \lim_{C\to\infty} \limsup_{n\to\infty} \sup_{r\in \mathcal{C}}[P\{\Vert\hat{r}_n(.|\tilde{\theta}) - r \Vert \geq C\delta_n\} + P\{\Vert\check{r} - r \Vert \geq C\delta_n\}]=0,
 \end{equation}
 subject to existence of $\tilde{\theta}$.

\begin{theorem}
If (3.4) holds for regression functions from $\mathcal{C}$ then
\begin{equation}
     \lim_{C\to\infty} \limsup_{n\to\infty} \sup_{r \in \mathcal{C}}P\{\Vert\hat{r}_n(.|\hat{\theta}) - r \Vert \geq C\delta_n\}=0.
\end{equation}
\end{theorem}
Theorem 2 states that $\hat{r}_n(.|\hat{\theta})$ and $\check{r}$ converge to $r$ uniformly in $\mathcal{C}$.

Let $X_1,X_2,...,X_n$ be iid random variables with probability density function (pdf) $f(x)$ and $\hat{g}(x)$ be its kernel based density function estimator
\begin{equation}
    \hat{g}(x)=\frac{1}{nh}\sum_{i=1}^{n}K(\frac{x-X_i}{h})\nonumber
\end{equation}
and $g(x)= E_f\hat{g}(x)$.

Suppose that (c) - (d) hold for $\phi_K$ and $\phi_{q}$, Fourier transforms for $K$ and $q=r\cdot g$, respectively,  
\begin{enumerate}
\item [(c)] $\phi_K{(t)^{-1}}=1+\sum_{j=1}^{k} c_jt^{2j},$ $c_1,...,c_k$ are real numbers;
\item [(d)] $\int|\phi_{q}(t)||t|^{2k}dt<\infty;\nonumber$
\item [(e)] For constants $C_1,...,C_5>0$, and  $j=1,...,k$ the derivatives ${q}^{(2j)}$ exist, $|{q}^{(2j)}(x)|\leq C_1$ and $\int|{q}^{(2j)}(x)|\leq C_1$, and either
\begin{enumerate}
    \item [(1)] $|{q}^{(2j)}(x)/{r \cdot f}(x)|\leq C_1$ for all $x$, or
    \item [(2)]  $|{q}^{(2j)}(x)/{r \cdot f}(x)|\leq C_1(1+|x|)^{C_2}$ for all $x$ and $P(|X|\geq x)\leq C_3 \exp(-C_4x^{C_5}),\; x>0$.
\end{enumerate}
\item [(f)] Under assumption (e)-(1), $\delta_n \rightarrow 0$ as $n \rightarrow \infty$, so that $n^{1/2}\delta_n \rightarrow \infty$ , and  under assumption (e)-(2), $n^{1/2}log(n)^{-C_2 /4C_5}\delta_n \rightarrow \infty$, where $C_2$ and $C_5$ are defined in (e)-(2).
\end{enumerate}
Assumption (c)-(f) are reasonable and considered in tilted density function estimation in \cite{doosti2016making}. 
It is anticipated that $\Vert\hat{r}_n(.|\hat{\theta}) - {r} \Vert =O_p(\delta_n)$, where $\delta_n$ converges to 0 slower than $n^{-1/2}$ as shown in Theorem 3. Next we formulate the assumption using the first term of the expression in the left hand side  of  (3.4) 
\begin{equation}
     \lim_{C\to\infty} \limsup_{n\to\infty} \sup_{r \in \mathcal{C}} P_r(\Vert\check{r} - r \Vert > C\delta_n)=0.
\end{equation}
\begin{theorem}
Let (a)-(f) be valid and $\hat{\theta}$ be defined in (3.1).
\begin{enumerate}
    \item[\textbf{I.}] If in addition $\Vert\check{r} - r \Vert =O_p(\delta_n)$  then $\Vert\hat{r}_n(.|\hat{\theta}) - r \Vert =O_p(\delta_n).$ 
    \item[\textbf{II.}] If the assumptions in (e) hold uniformly for $q \in \mathcal{C}$ and obtain (3.6) then (3.5) is valid.
\end{enumerate}
\end{theorem}
The proof of Theorem 3 is given in the Appendix. 
\section{Simulation study}
We present the results of simulation study of the performance of tilted estimators in various settings. Data were   generated using exponential function and sin regression functions with normal and uniform design distributions. Four samples  of sizes $n=(60, 100, 200, 1000)$ were sampled from the population with regression errors which had standard deviations $\sigma=(0.3, 0.5, 0.7, 1, 1.5, 2)$. In the setting for each set of $\sigma$ and $n$  we generated 500 data sets. The Median Integrated Squared Error (MISE) was estimated using the Monte Carlo method. The leave-one-out cross validation score from (1.5) was employed to choose the optimal bandwidths for Nadaraya-Watson and local linear estimators,  \cite{wasserman2006all}. For an infinite order flat-top kernel estimator, bandwidth was selected using the rule of thumb introduced by McMurry and Politis \cite{mcmurry2004nonparametric} as part of 'iosmooth'. The bandwidth parameters for tilted estimators were estimated using our suggested procedure.  The exponential  regression function $r_1(x)=x+4exp(-2x^2)/\sqrt{2\pi}$. The design densities were taken to be uniform on $[-2,2]$ and $N(0,1)$. The Integrated Squared Error (ISE) was calculated over the interval $[-2,2]$. The sin regression function $r_2(x)=sin(4\pi x)$ was paired with the uniform design density on $[0,1]$. The ISE has been calculated over  $[0,1]$ and $[0.15,0.85]$, the latter is chosen for addressing  boundary effect.
\clearpage

In Table \ref{tab:table-1} we provide  the MISEs for the proposed estimators, the comparator estimator, and the conventional estimators. Data were generated using $r_1(x)$ regression function along with normal design density and normal distribution for the error term. It is evident that for moderate sample size (n=200) and a large sample size (n=1000) and for the medium standard deviation (0.5 and 0.7), the tilted estimators outperformed other estimators.  Moreover, for larger sample size, as standard deviation of error terms increases, the MISE of tilted estimators decreases.
\begin{table}[h!]
\caption{\label{tab:table-1} MISE for Infinite Order (IO) estimator with the trapezoidal kernel, Nadaraya-Watson (NW) estimator, standard local linear (LL) estimator, tilted NW estimator with 4 (NW p4) and 10 (NW p10) weighting nodes, tilted LL estimator with 4 (LL p4) and 10 (LL p10) weighting nodes, Exponential regression function and normal design density function. In each row the minimum MISE is highlighted in bold.}
\begin{center}
\begin{tabular}{|c|c|c|c|c|c|c|c|c|} 
\hline
$n$ & $\sigma$ & IO & NW & LL & NW p4 &NW p10 & LL p4 & LL p10\\
\hline
\multirow{6}{2em}{60} & 0.3	& 0.2070 &	0.0861 & $\boldsymbol{0.0742}$ &	0.1422 & 0.1676 &	0.1194	& 0.1566 \\ 
& 0.5 &	0.2771 &	0.1630 &	$\boldsymbol{0.1486}$ & 0.2152 &	0.2294 &	0.1930 & 0.2203\\ 
& 0.7&	0.3755	& 0.2710 &	$\boldsymbol{0.2432}$ &	0.3039 &	0.3275 &	0.3009 &	0.3223  \\ 
& 1 & 0.5888 &	0.4626	& $\boldsymbol{0.4110}$ &0.5055 & 0.5597 &	0.5042 &	0.5261\\ 
&1.5&	1.1451&	0.9115&	$\boldsymbol{0.8514}$&	0.9576&	1.0860& 0.9590	& 1.0620\\ 
&2&	1.9295&	1.5667&	$\boldsymbol{1.4383}$&	1.6168&	1.7774&1.6119 &	1.7292\\
\hline
\multirow{6}{2em}{100} & 0.3&	0.0814&	0.0571&	$\boldsymbol{0.0462}$&	0.0609&	0.06478&0.05156&	0.0595\\
&0.5&	0.1382&	0.1044&	$\boldsymbol{0.0893}$&	0.1131&	0.1168&	0.1055 &	0.1107\\
&0.7&	0.2233&	0.1649&	$\boldsymbol{0.1454}$&	0.1879&	0.1903&	0.1744	& 0.1846\\
&1&	0.3955&	0.2842&	$\boldsymbol{0.2534}$&	0.3343&	0.3616 & 0.3166 & 0.3480\\
&1.5&	0.8122&	0.5757&	$\boldsymbol{0.5228}$&	0.6764&	0.7340&	0.6802	& 0.7217\\
&2&	1.4227&0.9750&$\boldsymbol{0.8954}$&1.1569&1.3134&	1.1310 &	1.2404
\\
\hline
\multirow{6}{2em}{200} &0.3&	0.04982&	0.02886&	$\boldsymbol{0.0267}$&	0.03518&	0.0435&	0.0293 & 0.0377\\
&0.5&	0.0740&	0.0632&	0.0553&	0.0589&	0.0640&	$\boldsymbol{0.0552}$& 0.05780\\
&0.7&	0.1071&	0.1092&	$\boldsymbol{0.0857}$ & 0.0907&	0.0936&	0.08777 &	0.0867\\
&1&	0.1743&	0.2046&	$\boldsymbol{0.1407}$&	0.1526&	0.1568&	0.1528 & 0.1527\\
&1.5&	0.3395&	0.43460&	$\boldsymbol{0.2751}$&	0.3022&	0.3232&	0.3007 &	0.3165\\
&2&	0.5732&	0.7590&	$\boldsymbol{0.4676}$&	0.5068&	0.5470&	0.4991 & 0.5310\\
\hline
\multirow{6}{2em}{1000} &0.3&	0.02481&	$\boldsymbol{0.0078}$&	0.0121&	0.0127&	0.0222&	0.0101 & 0.0187\\
&0.5&	0.0293	&0.0173&	0.0197&	0.0194&	0.0251&	$\boldsymbol{0.0166}$	&	0.0225
\\
&0.7&	0.0353&	0.0287&	0.0288&	0.02748&	0.0304&	$\boldsymbol{0.0247}$&	0.0277
\\
&1&	0.0502&	0.0494&	0.0466&	0.0443&	0.0435&	$\boldsymbol{0.0420}$&	0.04245\\
&1.5&	0.0825&	0.08977&	0.0812&	0.0799&	0.0776&0.0814&	$\boldsymbol{0.0752}$
\\
&2&	0.1279&	0.1442&	0.1254&	0.1292&	$\boldsymbol{0.1243}$&	0.1299	& 0.1246
\\
\hline
\end{tabular}
\end{center}
\end{table}
\clearpage

In Table \ref{tab:table-2}, we provide the MISEs for simulated data using the exponential regression function
$r_1(x)$  
With the uniform design density and the random normal error term.  For fixed sample size, as the standard deviation increases, the tilted estimators otperform others. Although, for large sample sizes, the conventional estimators tend to perform better than tilted estimators. For smaller sample sizes and the moderate standard deviation levels, the tilted N-W estimator remains superior to the conventional estimators  at some extent.

\begin{table}[h!]
\caption{\label{tab:table-2} MISE for Infinite Order (IO) estimator with the trapezoidal kernel, Nadaraya-Watson (NW) estimator, standard local linear (LL) estimator, tilted NW estimator with 4 (NW p4) and 10 (NW p10) weighting nodes, tilted LL estimator with 4 (LL p4) and 10 (LL p10) weighting nodes, exponential regression function and uniform design density. In each row the minimum MISE is highlighted in bold.}
\begin{center}
\begin{tabular}{|c|c|c|c|c|c|c|c|c|} 
\hline
$n$ & $\sigma$ & IO & NW & LL & NW p4 &NW p10 & LL p4 & LL p10\\
\hline
\multirow{6}{2em}{60} &0.3&0.1559&	$\boldsymbol{0.0663}$&	0.0566&	0.1308&	0.1529&	0.1237 & 0.1470\\
&0.5& 0.1980& $\boldsymbol{0.1398}$&0.1362&	0.1724&	0.1953&	0.1690&	0.1901\\ 
& 0.7&	0.2515&	$\boldsymbol{0.2152}$&	0.2098&	0.2316&	0.2492&	0.2406&	0.2433\\ 
&1&	0.3588&	0.3697&	0.3599&	$\boldsymbol{0.3418}$&	0.3650&	0.3691&	0.3608
\\ 
&1.5&	0.6530&	0.6281&	0.6821&	$\boldsymbol{0.6197}$&	0.6520&	0.6676 &0.6692
\\ 
&2&	1.0524&	0.9892&	2.2171&	$\boldsymbol{0.9871}$&	1.0597&	1.0287 & 1.0287\\
\hline
\multirow{6}{2em}{100} & 0.3&	0.1195&	0.0442&	$\boldsymbol{0.0360}$&	0.1034&	0.1191&	0.0982	&0.1156\\
&0.5&	0.1432&	0.0914&	$\boldsymbol{0.0809}$&	0.1253&	0.1426&	0.1218&	0.1372
\\
&0.7&	0.1781&	0.1443&	$\boldsymbol{0.1376}$&	0.1607&	0.1766&	0.1581	&0.1695
\\
&1&	0.2490&	0.2324&	0.2438&	0.2305&	0.2469&	0.2497	& 0.2460\\
&1.5&	0.4165&	0.4366&	0.5221&	$\boldsymbol{0.4041}$&	0.4144&	0.4373&	0.4198
\\
&2&	0.6487&	0.6780&	0.8402&	$\boldsymbol{0.6107}$&	0.6371&	0.6619&	0.6401\\
\hline
\multirow{6}{2em}{200} &0.3&0.0991&	0.0232&	$\boldsymbol{0.0209}$&	0.0891&	0.0997&	0.0833	&0.0972
\\
&0.5&	0.1089&	0.0470&	$\boldsymbol{0.0441}$&	0.0993&	0.1086&	0.0944&	0.1063\\
&0.7&	0.1256&	0.0822&	$\boldsymbol{0.0757}$&	0.1172&	0.1253&	0.1107&	0.1228\\
&1&	0.1577&	$\boldsymbol{0.1299}$&	0.1351&	0.1533&	0.1589 &	0.1554	& 0.1590\\
&1.5&	0.2401&	0.2542&	0.3349&	$\boldsymbol{0.2386}$&	0.2416&	0.2587 & 0.2426\\
&2&	0.3568&	0.3878&	0.4218&	$\boldsymbol{0.3464}$&	0.3534&	0.3938 & 0.3573\\
\hline
\multirow{6}{2em}{1000} &0.3& 0.0801&	0.0058&	$\boldsymbol{0.0056}$&	0.0776&	0.0800&	0.0724	&0.0797\\
&0.5&	0.0823&	$\boldsymbol{0.0125}$&	0.01286&	0.0790&	0.0825&	0.0718	&0.0819\\
&0.7&	0.0853&	0.0207&	$\boldsymbol{0.0206}$&	0.0801&	0.0845&0.07170&	0.0843
\\
&1&	0.0922&	$\boldsymbol{0.0356}$&	0.0359&	0.0830&	0.0917&	0.0728	& 0.0895
\\
&1.5&	0.1080&	0.0716&	$\boldsymbol{0.0670}$&	0.0972&	0.1074&	0.0911&	0.1041
\\
&2&	0.1294&	0.1286&	$\boldsymbol{0.1056}$&	0.1209&	0.1308&	0.1235&0.1271\\
\hline
\end{tabular}
\end{center}
\end{table}
\clearpage
Table \ref{tab:table-3} presents the MISEs for the simulated data using sin function with uniform design density and normal random errors. For  fixed sample size and moderate standard deviations 0.5 and 0.7, the tilted estimators perform better than conventional estimators. For sample size n=1000 with and increasing standard deviation, the tilted estimators demonstrate better performance over others.
\begin{table}[h!]
\caption{\label{tab:table-3} MISE for Infinite Order (IO) estimator with the trapezoidal kernel, Nadaraya-Watson (NW) estimator, standard local linear (LL) estimator, tilted NW estimator with 4 (NW p4) and 10 (NW p10) weighting nodes, tilted LL estimator with 4 (LL p4) and 10 (LL p10) weighting nodes, sin regression function, uniform design density, edges included. In each row the minimum MISE is highlighted in bold.}
\begin{center}
\begin{tabular}{|c|c|c|c|c|c|c|c|c|} 
\hline
$n$ & $\sigma$ & IO & NW & LL & NW p4 &NW p10 & LL p4 & LL p10\\
\hline
\multirow{6}{2em}{60} &0.3&0.04034&	0.0286&	0.0234&	0.0315&	0.0331&	$\boldsymbol{0.0228}$ &	0.0264
\\ 
&0.5&0.0663&	0.0638&	0.0570&	0.0585&	0.0597&	$\boldsymbol{0.0507}$&	0.0534\\ 
& 0.7&0.1085&$\boldsymbol{0.0840}$&	0.1329&	0.0971&	0.0969&	0.0897 & 0.090\\ 
&1&	0.1958&	0.1703&	$\boldsymbol{0.1424}$&	0.1731&	0.1749&	0.1714 & 0.1692
\\ 
&1.5&0.4036&$\boldsymbol{0.2410}$&	0.2652&	0.3616&	0.3718&	0.3600&	0.3604\\ 
&2&0.7003&	$\boldsymbol{0.3739}$&	0.3958&	0.6198&	0.6463&	0.6203 &0.6307\\
\hline
\multirow{6}{2em}{100} & 0.3&0.0222&	0.0166&	$\boldsymbol{0.0129}$&	0.0193&	0.0197&	0.0130	&0.0152
\\
&0.5&0.0371&	0.0326&	0.0286&	0.0348&	0.0349&	$\boldsymbol{0.0282}$ & 0.0308
\\
&0.7&0.0595&0.0533&	$\boldsymbol{0.0498}$&	0.0560&	0.0554&	0.0506 &	0.0517\\
&1&0.1054&	0.0936&	$\boldsymbol{0.0868}$&	0.1008&	0.1004&	0.0989&	0.09720
\\
&1.5&0.2183&$\boldsymbol{0.1541}$&	0.1728&	0.2078&	0.2072&	0.2067 & 0.2032\\
&2&	0.3784&	$\boldsymbol{0.2869}$&	0.2985&	0.3520&	0.3580&	0.3573 &	0.3549\\
\hline
\multirow{6}{2em}{200} &0.3&0.0123&	0.0093&	0.0070&	0.0112&	0.0114&	$\boldsymbol{0.0067}$&	0.0080\\
&0.5&0.0191&0.0190&	0.0161&	0.0195&	0.0194&	$\boldsymbol{0.0143}$ &	0.0150\\
&0.7&0.0299&0.0306&	0.0273&	0.0308&	0.0296&	$\boldsymbol{0.0251}$&	0.0260\\
&1&	0.0522&	0.0492&	0.0494&	0.0533&	0.0517&	0.0484 &	$\boldsymbol{0.0471}$\\
&1.5&0.1073&$\boldsymbol{0.0885}$&	0.1043&	0.1046&	0.1048&	0.1044 &	0.1011\\
&2&	0.1830&	$\boldsymbol{0.1334}$&	0.1432&	0.1770&	0.1789&	0.1812	&0.1752\\
\hline
\multirow{6}{2em}{1000} &0.3&0.0056&0.0023&	0.0019&	0.0046&	0.0051&	$\boldsymbol{0.0018}$&	0.0027
\\
&0.5&0.0070&0.0049&	0.0042&	0.0065&	0.0066&	$\boldsymbol{0.0034}$&	0.0043
\\
&0.7&0.0089&0.0082&	0.0071&0.0091&	0.0091&	$\boldsymbol{0.0056}$&	0.0064
\\
&1&0.0132&0.0145&0.0127&0.01402&0.01346&$\boldsymbol{0.0101}$&	0.0108\\
&1.5&0.0234&0.0254&	0.0241&	0.0250&	0.0242&$\boldsymbol{0.0214}$& 0.0212\\
&2&	0.0376&0.0424&0.0386&0.0404&0.0388&	0.0369	& $\boldsymbol{0.0358}$\\
\hline
\end{tabular}
\end{center}
\end{table}
\clearpage
For studying boundary effect the results provided in Table \ref{tab:table-3} and \ref{tab4}  were evaluated under the identical experimental specifications except the MISEs in Table \ref{tab4} were evaluated over [0.15,0.85]. According to the results when the sample size and standard deviation increased, the tilted estimators demonstrated improved performance. 
\begin{table}[h!]
\caption{\label{tab4} MISE for Infinite Order (IO) estimator with the trapezoidal kernel, Nadaraya-Watson (NW) estimator, standard local linear (LL) estimator, tilted NW estimator with 4 (NW p4) and 10 (NW p10) weighting nodes, tilted LL estimator with 4 (LL p4) and 10 (LL p10) weighting nodes, sin regression function, uniform design density, edges excluded. In each row the minimum MISE is highlighted in bold.}
\begin{center}
\begin{tabular}{|c|c|c|c|c|c|c|c|c|} 
\hline
$n$ & $\sigma$ & IO & NW & LL & NW p4 &NW p10 & LL p4 & LL p10\\
\hline
\multirow{6}{2em}{60} &0.3&0.0261&	0.0182&	0.0143&	0.0201&	0.0213&	$\boldsymbol{0.0142}$ &	0.0184\\ 
&0.5&0.0461&0.0421&	0.0327&	0.03747&0.0394&	$\boldsymbol{0.0319}$&	0.0365\\ 
&0.7&0.0740&$\boldsymbol{0.0498}$&	0.0791&	0.0628&	0.0638&	0.05783	&0.0618\\ 
&1&0.1315&0.1130&$\boldsymbol{0.0796}$&0.1189&	0.1209&	0.1095&	0.1184
\\
&1.5&0.2736&$\boldsymbol{0.1461}$&	0.1468&	0.2519&	0.2610&	0.2420&	0.2513
\\ 
&2&	0.4747&	0.2646&	$\boldsymbol{0.2208}$&	0.4271&	0.4508&	0.4215	&0.4406
\\
\hline
\multirow{6}{2em}{100} & 0.3&0.0132&0.0101&	$\boldsymbol{0.0078}$&	0.0108&	0.0109&	0.0082	&0.0096
\\
&0.5&0.0233&0.0194&	$\boldsymbol{0.0168}$&	0.0201&	0.0207&	0.0173	&0.0197
\\
&0.7&0.0380&0.0307&	$\boldsymbol{0.0285}$&	0.0340&	0.0346&	0.0308 &	0.0331
\\
&1&0.0688&	0.0562&	$\boldsymbol{0.0493}$&	0.0626&	0.0641&	0.0587 &	0.0626
\\
&1.5&0.1459&$\boldsymbol{0.0912}$&	0.0943&	0.1319&	0.1365&	0.1270&	0.1349\\
&2&0.2536&0.1750&$\boldsymbol{0.1584}$&0.2304&	0.2391&	0.2228	&0.2362
\\
\hline
\multirow{6}{2em}{200} &0.3&0.0063&	0.0054&	0.0044&	0.0054&	0.0054&	$\boldsymbol{0.0040}$&	0.0049
\\
&0.5&0.0112&0.0121&	0.0095&	0.0100&	0.0102&	$\boldsymbol{0.0084}$&	0.0097\\
&0.7&0.0183&0.01733&0.0164&	0.0162&	0.0166&	$\boldsymbol{0.0149}$&	0.0164
\\
&1&0.0330&	0.0296&	0.0300&	0.0298&	0.0307&	$\boldsymbol{0.0288}$&	0.0302
\\
&1.5&0.0679&$\boldsymbol{0.0531}$&	0.0627&	0.0636&	0.0653&	0.0610&	0.0645\\
&2&	0.1160&	$\boldsymbol{0.0816}$&	0.0834&	0.1116&	0.1125&	0.1094&0.1132
\\
\hline
\multirow{6}{2em}{1000} &0.3&0.0013&0.0014&	0.0012&	0.0011&	0.0011&	$\boldsymbol{0.0009}$&	0.0010
\\
&0.5&0.0022&0.0028&	0.0026&	0.0019&	0.0018&	$\boldsymbol{0.0017}$&	0.0018
\\
&0.7&0.0035&0.0045&	0.0044&	0.0032&	0.0033&	$\boldsymbol{0.0029}$&	0.0031
\\
&1&	0.0064&	0.0089&	0.0077&	0.0060&	0.0062&	$\boldsymbol{0.0057}$&	0.0060
\\
&1.5&0.0136&0.0152&	0.0144&	0.0126&	0.0131&	$\boldsymbol{0.0123}$&	0.0130\\
&2&0.0236&	0.0247&	0.0229&	0.0219&	0.0229&	$\boldsymbol{0.0218}$&	0.0225
\\
\hline
\end{tabular}
\end{center}
\end{table}
\clearpage
It is a known fact that the performance of  Nadaraya-Watson estimator deteriorates near edges,   \cite{martinez2007computational}. This effect is often referred to as a boundary problem. 
The results presented in Table 3 and  4, illustrate that the for scenarios $(n=60,\; \sigma=0.5),$  $(n=200,\; \sigma=0.7)$ and $(n=1000,\; \sigma=\{1, 1.5, 2\})$ the tilted Nadaraya-Watson estimator outperformed its classical counterpart.
From the boxplots in Figure \ref{boxblot} it is evident that the tilted estimators have smaller median ISEs. The extreme values of the ISEs for the tilted estimators are smaller than these of the conventional estimators. Similarity between the ISE distributions and their spreads of the IO  and tilted estimators can also be seen in Figure \ref{boxblot}.  
\begin{figure}[htp]
    \centering
    \includegraphics[width=12cm]{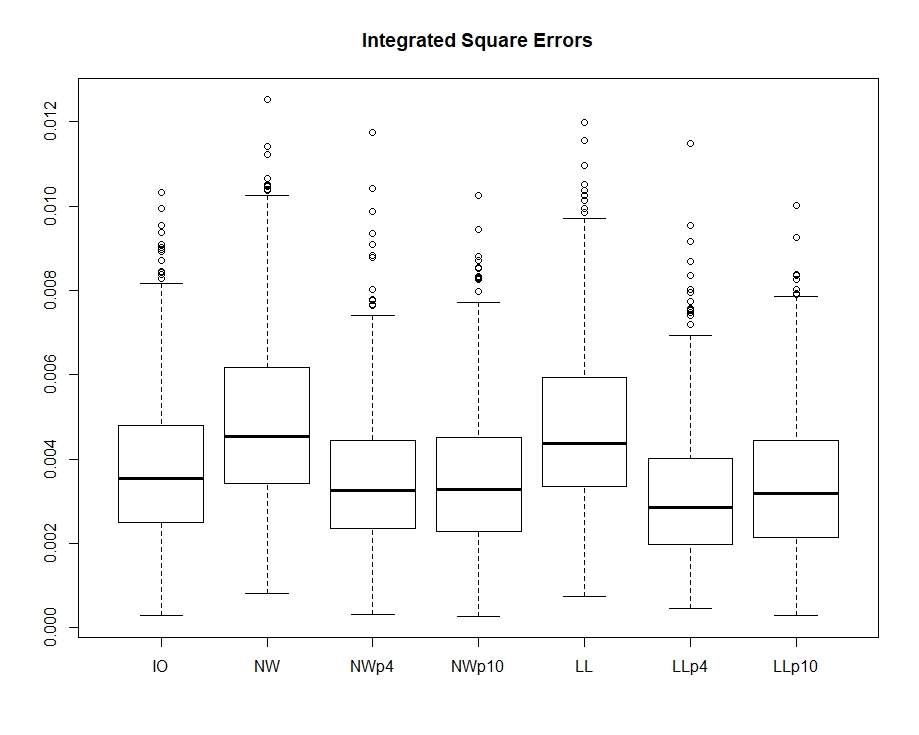}
    \caption{Boxplots of Integrated Square Errors (ISE) for Infinite Order (IO) estimator with the trapezoidal kernel, Nadaraya-Watson (NW) estimator, standard local linear (LL) estimator, tilted NW estimator with 4 (NW p4) and 10 (NW p10) weighting nodes, tilted LL estimator with 4 (LL p4) and 10 (LL p10) weighting nodes, sin regression function, edges excluded, $n=1000$ and $\sigma=0.7$ }
    \label{boxblot}
\end{figure}
\clearpage
 In this simulation study for carrying out the MISE analysis,  we had 500 replications done using Monte Carlo method. For MISE evaluation at each combination of an estimator, a function, a standard deviation  and for a fixed sample size we had to solve the optimization problem. For the numerical implementation, we used the parallel computing technique in R  facilitated through 'snow', 'doparallel', and 'foreach' packages.

\section{Real data}
In this section, we study the performance of tilted estimators  in the real data environment. 
\subsection{ COVID-19 Data }

The tilted N-W estimator along with two other kernel-based estimators are being used for a curve fitting to the COVID-19 data. We shall apply the tilted N-W estimator approach to daily confirmed new cases and number of daily death for 12 countries including Iran, Australia, Italy, Belgium, Germany, Spain, Brazil, United Kingdom, Canada, Chile, South Africa and United States of America, from 23 February 2020 to 28 October 2020, downloaded from
 \href{https://www.ecdc.europa.eu/en/publications-data/download-todays-data-geographic-distribution-covid-19-cases-worldwide}{https://www.ecdc.europa.eu}. The logarithmic  transformation has been applied, and when the number of deaths or new confirmed cases were zero we altered these observations by a positive value eliminating the associated singularity issue. The optimal bandwidth for each Nadaraya-Watson estimator was found through minimization of relevant cross-validation function, \cite{wasserman2006all}, at the same time we kept the bandwidth fixed for an infinite order flat-top kernel (IO) estimator which was found using Mcmurry and Politis’ rule of thumb, \cite{mcmurry2004nonparametric}. 
Along with the tilted N-W estimator, we applied the NW, and IO estimators. The tilted NW estimator performed the best in terms of the Mean Square Errors (MSE).
Table \ref{msecase} and \ref{msedeath} provide the MSE for each estimator for the confirmed new cases and number of deaths. In terms of minimising the MSE, the tilted NW estimator ranked first, followed by IO and N-W estimators. Slightly, improved performance of the tilted NW estimator is attributed to the lower MSE components at the edges versus other kernel-based regression function estimators which are generally known for so-called "edge effect", \cite{hall1991geometrical}.
Figures \ref{case} and \ref{death} shows curve fit to the daily confirmed cases and daily deaths respectively. 

\begin{figure}[htp]
    \centering
    \includegraphics[width=22cm,
  height=20cm,
  keepaspectratio,
 angle=90]{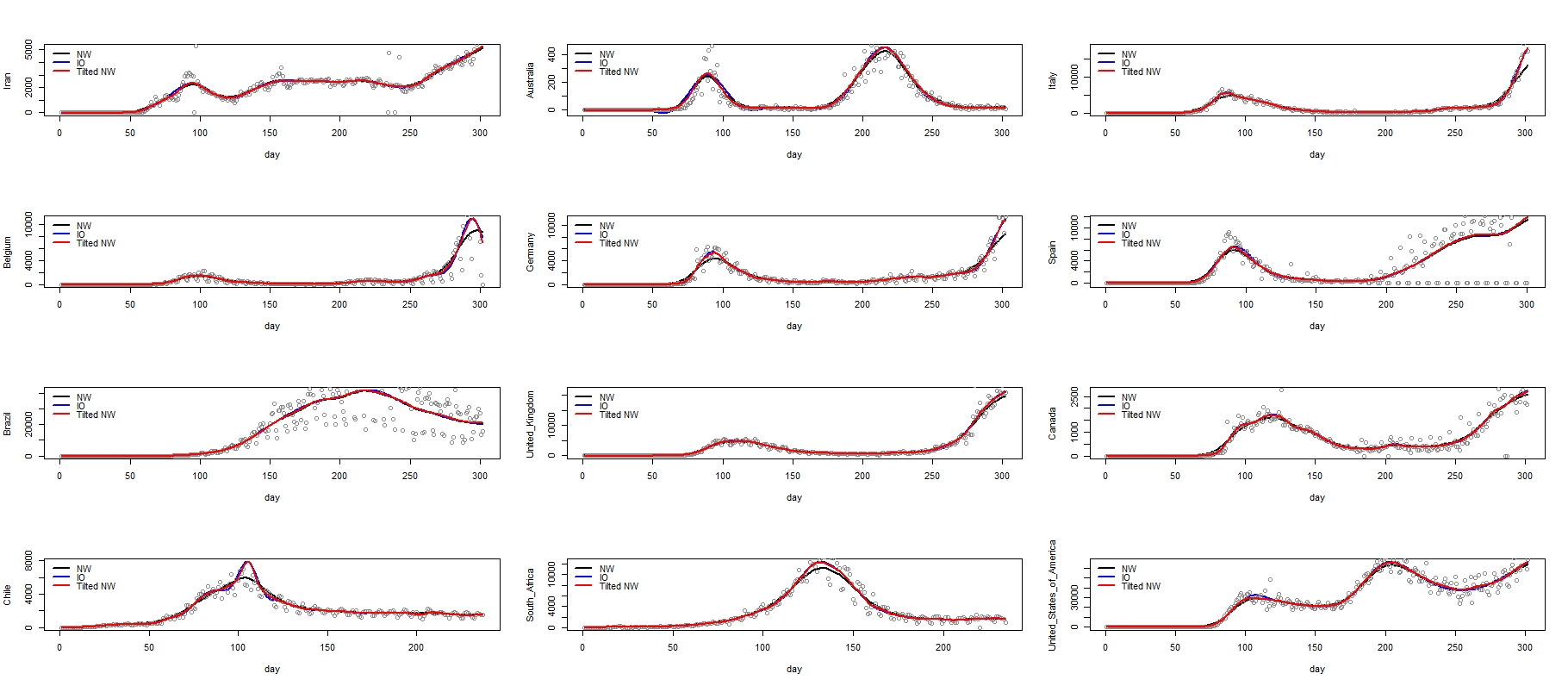}
    \caption{Daily confirmed cases for 12 countries including Iran, Australia, Italy, Belgium, Germany, Spain, Brazil, United Kingdom, Canada, Chile, South Africa and United States of America }
    \label{case}
\end{figure}
\clearpage

\begin{table}[h!]
\caption{\label{msecase} COVID-19 Daily Confirmed Cases: MSE for Nadaraya-Watson (NW), Infinite Order (IO), and tilted (NW p4) estimators. In each row the minimum MSE is highlighted in bold.}
\begin{center}
     \setlength{\tabcolsep}{2pt}
     \small
\begin{tabular}{|c|c|c|c|} 
\hline
Country & IO & NW & NW p4 \\
\hline
 Iran &324
  & 	326	 & $\boldsymbol{305}$ \\
\hline
Australia & 5.48	
 & 5.41	 & $\boldsymbol{5.01}$\\
\hline
Italy & 445	
 &1590 &$\boldsymbol{372}$\\
\hline
Belgium &  2357	
 &  3071	& $\boldsymbol{2232}$\\
\hline
Germany & 641	
  & 1029	 & $\boldsymbol{594}$\\
\hline
Spain & 41621	
 &  41754	 & $\boldsymbol{41241}$\\
\hline
Brazil &108590	
  & 108246	 & $\boldsymbol{107781}$\\
\hline
United Kingdom & 1789	
 & 2083	 & $\boldsymbol{1737}$\\
\hline
Canada & 175
 & 	183	 & $\boldsymbol{173}$\\
\hline
Chile & 6040	
 & 6685 & $\boldsymbol{	5991}$\\
\hline
South Africa &1158
 & 	1403	 &$\boldsymbol{1133}$\\
\hline
 United States of America & 43910
 &46512 & $\boldsymbol{	41694}$\\
 \hline
\end{tabular}
\end{center}
\end{table}

\begin{table}[h!]
\caption{\label{msedeath} COVID-19 Death: MSE for Nadaraya-Watson (NW), Infinite Order (IO), and tilted (NW p4) estimators. In each row the minimum MSE is highlighted in bold.}
\begin{center}
     \setlength{\tabcolsep}{2pt}
     \small
\begin{tabular}{|c|c|c|c|} 
\hline
Country & IO & NW & NW p4 \\
\hline
 Iran &1.36
  & 	1.39	 & $\boldsymbol{1.32}$ \\
\hline
Australia &	0.0225	

 & 	0.0223	 & $\boldsymbol{0.0222}$\\
\hline
Italy & 1.97	

 &3.66 &$\boldsymbol{	1.89}$\\
\hline
Belgium &  0.0699

 &  	0.3077		& $\boldsymbol{0.0690}$\\
\hline
Germany & 0.71
	
  & 	0.82 & $\boldsymbol{	0.68	}$\\
\hline
Spain & 37.74	
	
 &  37.92	 & $\boldsymbol{36.77}$\\
\hline
Brazil & 63.99
	
  & 	64.21		 & $\boldsymbol{63.86}$\\
\hline
United Kingdom & 13.57
	
 & 		14.50	 & $\boldsymbol{13.01}$\\
\hline
Canada & 0.40

 & 		0.47		 & $\boldsymbol{0.39}$\\
\hline
Chile & 9.10
	
 & 	9.73	 & $\boldsymbol{8.95	}$\\
\hline
South Africa & 3.32
 & 	3.50 &$\boldsymbol{3.22	}$\\
\hline
 United States of America & 201.39	
 
 & 205.98& $\boldsymbol{	197.76	}$\\
 \hline
\end{tabular}
\end{center}
\end{table}
\begin{figure}[htp]
    \centering
    \includegraphics[width=22cm,
  height=20cm,
  keepaspectratio,
 angle=90]{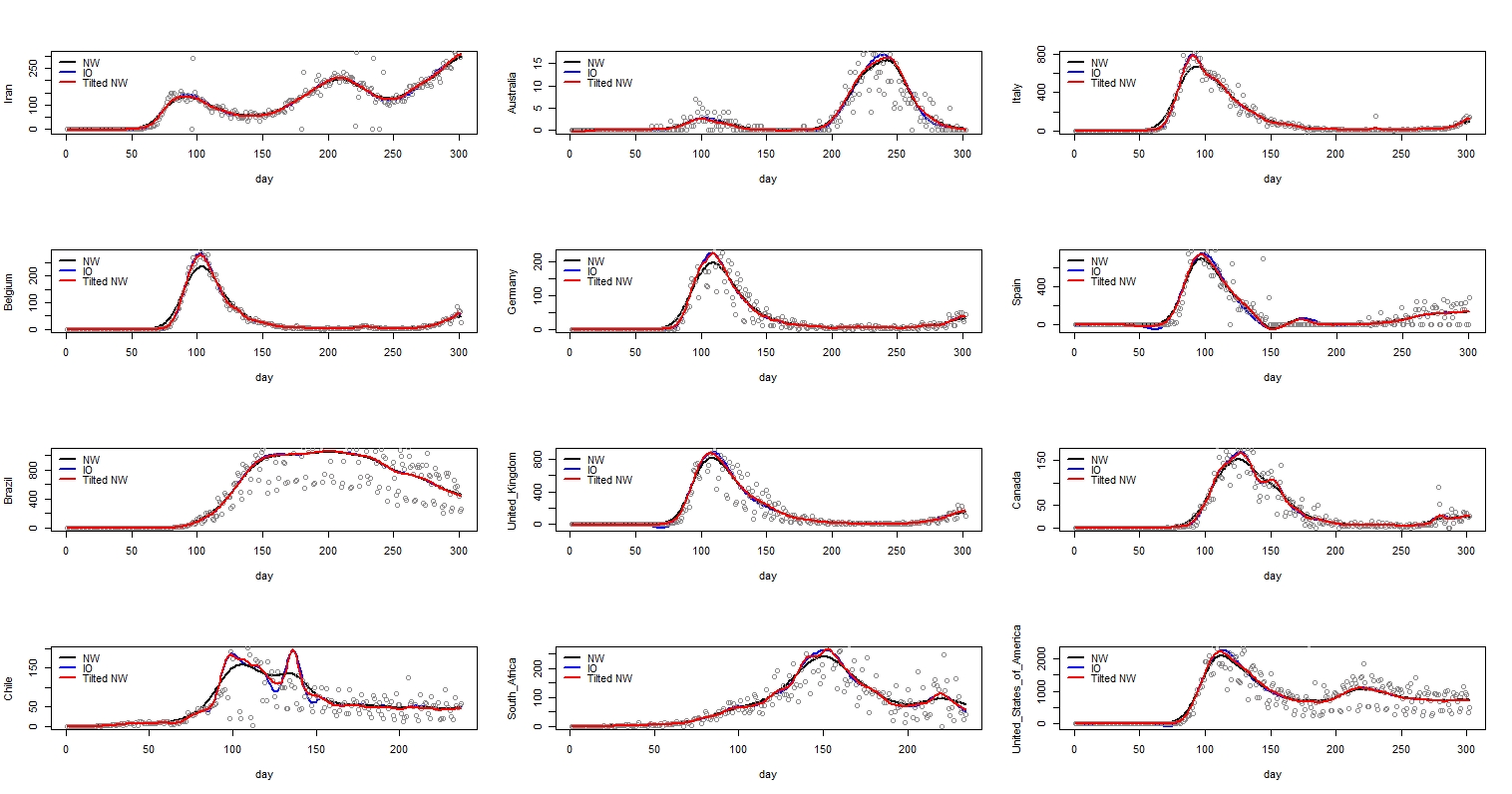}
    \caption{Daily deaths for 12 countries including Iran, Australia, Italy, Belgium, Germany, Spain, Brazil, United Kingdom, Canada, Chile, South Africa and United States of America }
    \label{death}
\end{figure}
\clearpage

\subsection{Dose-Response data}
The dose-response data refers to a study of phenylephrine effects on rat corpus cavernosum strips. This data first appeared in  Boroumand et al. \cite{boroumand2016application} where the dose-response curves to phenylephrine (0.1 $\mu M$ to 300 $\mu M$) were obtained by applying the robust four-parameter logistic (4PL) regression.
Here we have used a tilted smoother approach to  dose-response curve fitting. In terms of Mean Square Errors (MSEs) the tilted local linear estimators performed better than local linear, infinite order flat-top kernel estimators including the robust 4PL  model.
The fitted dose-response curves using the tilted local linear, local linear, infinite order flat-top kernel estimator, and 4PL model are plotted in Figure \ref{real data3}. The corresponding MSEs are listed in the caption of Figure \ref{real data3}.
\begin{figure}[htp]
    \centering
    \includegraphics[width=12cm]{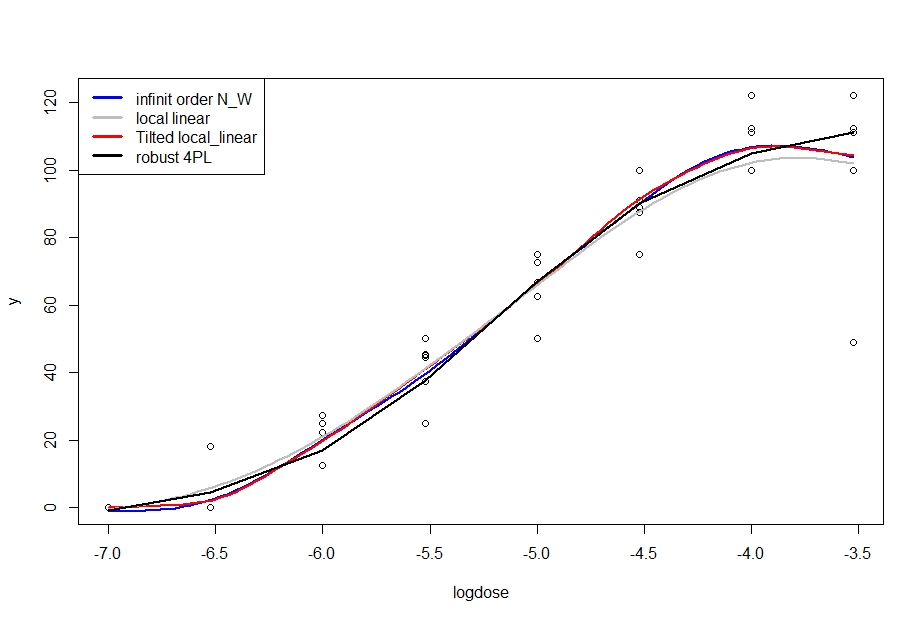}
    \caption{Dose-response curves: MSE for tilted local linear, local linear, infinite order flat-top kernel estimator and 4PL model are 95.1023, 95.3267, 95.53077 and 110.7539, respectively.}
    \label{real data3}
\end{figure}

The original  dose-response data contained the outliers and the standard 4PL model had a poor fit. Due to this we compared the performance of the tilted estimator with the robust 4PL model. The tilted local linear estimator outperformed the robust 4PL model in terms of MSE.

\section*{Acknowledgement}
This research was undertaken with the assistance of resources and services from the National Computational Infrastructure (NCI), which is supported by the Australian Government.
This research forms part of the first author’s PhD thesis approved by the ethics committee of Mashhad University of Medical Sciences with the project code 971017.
\medskip
\printbibliography
\clearpage

\begin{appendices}
\section{Proof of Theorem 3}
In this section we provide proof of Theorem 3 for tilted Nadaraya-Watson estimator as a form of tilted linear smoother from (2.3). The result for tilted local linear smoother can be proved analogously.

\begin{proof}
Let
$p_i=\frac{1}{n}\pi(X_i)$ where $\pi \geq0$ and is a smooth function, $\sum_{i=1}^{n}p_i=1$ which is equivalent  $\int\pi(X)f_X(x)dx=1$ for continuous $X$. For simplicity, we replace $f_X(x)$ by $f$ then

\begin{align}
    \sum_{i=1}^{n}p_i & =\frac{1}{n}\sum_{i=1}^{n}\pi(X_i) \nonumber\\
                      & =E\pi(X) + O_p(n^{-1/2}) \nonumber\\
                      & = \int\pi(X) fdx + O_p(n^{-1/2})\nonumber\\
                      & = 1+O_p(n^{-1/2}),\nonumber
\end{align}
thus for normalising  $p_is$ so that $\sum_{i}p_i=1$ we need to multiply the estimator (2.3) by $1+O_p(n^{-1/2})$. The factor $O_p(n^{-1/2})$ is negligibly small. 
We choose $\pi$ such that $\hat{r}_n(x|h,p)$ in (2.3)
is unbiased estimator for $r$, i.e
\begin{align}
E\hat{r}(x|h,p)=r.    \nonumber \\ 
\end{align}
From (A.1) we have
\begin{align}
    E \hat{r}(x)\hat{g}(x) & = \frac{1}{h} \sum_{i=1}^{n} E{p_i Y_i 
                            K(\frac{x-X_i}{h})} \nonumber\\
                           & = \frac{1}{nh} \sum_{i=1}^{n} EE\{Y_i \pi(X_i)K(\frac{x-X_i}{h})|X_i\} \nonumber\\
                           & = \frac{1}{h} \int_{-\infty}^{\infty} r(t) \pi(t)K(\frac{x-t}{h})f_X(t)dt,
\end{align}
where
\begin{equation}
    \hat{g}(x)=\frac{1}{nh}\sum_{i=1}^{n}K(\frac{x-X_i}{h})\nonumber
\end{equation}
and $g(x)= E\hat{g}(x)$.
It can be shown that the left-hand side of (A.2) is converging to $r(x)g(x).$ 

We have 
$$r(x)g(x)=\frac{1}{h} \int_{-\infty}^{\infty} r(t') \pi(t')K(\frac{x-t'}{h})f_X(t')dt',$$
multiplying both sides by $e^{-itx}$  and integrating over $x$, we deduce
$$\Phi_{rg}(t)=\frac{1}{h}\int_{-\infty}^{\infty}\int e^{-itx} r\pi f(t')K(\frac{x-t'}{h}) dx' dt',$$
by changing variable $\frac{x-t'}{h}=u$, we have

\begin{align}
   \Phi_{rg}(t) &= \int_{-\infty}^{\infty} e^{-itx} r\pi f(t')\Phi_k (t) dt' \nonumber\\
                &= \Phi_k(t) \Phi_{r \pi f}(t). \nonumber
\end{align}
\begin{align}
    \Phi_{r \pi f} &= \frac{\Phi_{rg}(t)}{\Phi_k(t)}, \nonumber\\
    \pi r f &= \frac{1}{2\pi}\int_{-\infty}^{\infty} e^{-itx}  \frac{\Phi_{rg}(t)}{\Phi_k(t)},\nonumber\\
    \pi(X) &= \frac{1}{2 \pi r f}\int_{-\infty}^{\infty} e^{-itx}  \frac{\Phi_{rg}(t)}{\Phi_k(t)} dt,
\end{align}

 if kernel $K$ holds the assumption (c) and $q=r\cdot g$ meet the assumption (d),
then
\begin{align}
    \pi &= 1+ \sum_{j=1}^{n} C_j (-h^2)^j \frac{rg^{(2j)}}{rf},
\end{align}
with $\pi$ from (A.3), then $\hat{r}_n$ in unbiased. Next we show that $\pi$ satisfies $0< \pi(X)<1$.

If the assumption (e) relaxed then there exist $C_6$ and $h_0 \geq 0$, for all $h$, $0\leq h\leq h_0$,  $\pi>0$ and $\sup \pi \leq C_6 < \infty$
then for unbiased $\hat{r}_n$

\begin{align}
    \int var \{\hat{r}_n(x|h,p)\}dx &\leq \frac{1}{nh^2} \int E\{\pi^2(X)K^2( \frac{x-X}{h})\}dx, \nonumber\\
                                    &\leq \frac{1}{nh} (\sup \pi)^2 \int K^2 dx, \nonumber\\
                     & = O\{nh^{-1}\}.
\end{align}
 So $MSE$ can be written as 
\begin{align}
    MSE\{\hat{r}_n(x|h,p)\} &= \int E \{\hat{r}_n(x|h,p) - r\}^2 dx,\nonumber \\
    &= O\{nh^{-1}\}.\nonumber
\end{align}
We recall that
\begin{enumerate}
\item [(f)] Under the assumption (e)-(1), $\delta_n \rightarrow 0$ as $n \rightarrow \infty$, so that $n^{1/2}\delta_n \rightarrow \infty$ , and  under assumption (e)-(2), $n^{1/2}log(n)^{-C_2 /4C_5}\delta_n \rightarrow \infty$, where $C_2$ and $C_5$ are defined in (e)-(2).
\end{enumerate}
Then under the assumption (e)-(1) and (f), we have that $n^{1/2} \delta_n \rightarrow \infty$ thus $n^{-1}=o(\delta^2_n)$. Consequently, there exists $h(n)\downarrow 0$ as $n \rightarrow \infty$ such that $(nh)^{-1}=O(\delta^2_n),$
 for some large $n,$ $h<h_0$   since $0\leq h\leq h_0$. Next, by replacing $O(\delta^2_n)$ in the right-hand side of (A.5), we have a new form of (A.5) which is true for specific choice of $\pi$ defined at (A.3), and considering $\tilde{\theta}=(h,p)$ in the case of (2.3): 
\begin{align}
    \lim_{C\to\infty} \limsup_{n\to\infty} P\{\Vert\hat{r}_n(.|\tilde{\theta}) - r \Vert \geq C\delta_n\}=0.
\end{align}
For this version of $p_i=n^{-1}\pi(X_i)$,  $\sum_{i}p_i=1$ does not satisfy. However, this issue can be fixed by normalisation similar to that done in the first paragraph of the proof.

Property (A.6) implies  part $\textbf{I}$ of Theorem 3 and part $\textbf{II}$ can be concluded under uniformity of (A.6) over $\mathcal{C}  $.

Under assumption (e)-(2) and  (A.4), $|rg^{(2j)}(x)/rf|\leq C_1(1+|x|)^{C_2}$ for $1\leq j\leq k$, defining $C_6= \max (|C_1|,...,|C_k|)$, we have for $0 \leq h \leq 1$
\begin{align}
    |\pi(X)-1|&= |\sum_{j=1}^{n} C_j (-h^2)^j \frac{rg^{(2j)}}{rf}|, \nonumber\\
              &\leq C_1C_6 K(1+|x|)^{C_2}h_2, \nonumber
\end{align}
so, if $\lambda_{1n}\rightarrow \infty$ and $\lambda_{1n}^{C_2}h^2\rightarrow 0$, then
\begin{align}
    \sup|\pi(X)-1|\rightarrow 0
\end{align}
it means that whenever $|X|\leq \lambda_{1n} $, $0< \pi(X)<1$.

In the first paragraph of the proof, we showed that $\pi(X)\geq 0$. Then we found an upper bound for $\pi(X)$ in (A.7) when $X \in [-\lambda_{1n}, \lambda_{1n}]$. Now, we want to show that the probability of $X$ being out of this interval is almost zero which means for all $X$, $0< \pi(X)<1$.

Assumption (e)-(2) implies that 
\begin{align}
    P(|X|\geq \lambda_{1n}) &\leq C_3 exp (-C_4 \lambda_{1n}^{C_5}).
\end{align}
Using (f), $n^{-1/2}(\log n)^{C_2/2C_5} \delta_n \rightarrow \infty$, or equivalently, $\delta_n^2= \lambda_{2n} n^{-1} (\log n)^{C_2/2C_5} \rightarrow \infty$, where $\lambda_{2n}$ exists. We choose $h$ so that $(nh)^{-1}=O(\delta_n^2)$; or for simplicity, $(nh^{-1})=\delta_n^2$, then $h=\lambda_{2n}^{-1}(\log n)^{-C_2/2C_5}$; let $\lambda_{1n}=\{ \lambda_{2n}^{2-\eta} (\log n)^{ C_5/C_2}\}^{1/C_2}$, where $\eta \in (0,2)$, so
\begin{align}
    \exp (-C_4 \lambda_{1n}^{C_5}) &= \exp (-C_4\lambda_{2n}^{(2-\eta) C_5/C_2 } \log n), \nonumber\\
                                 &= O(n^{-C}), \nonumber
\end{align}
for all $C>0$. Therefore by (A.8),
\begin{align}
    P(|X|\geq \lambda_{1n}) &= O(n^{-C}),\nonumber
\end{align}
for all $C>0$.
\end{proof}

\end{appendices}
\end{document}